\documentclass[final,3p]{elsarticle}

\usepackage{graphics}
\usepackage[OT4]{fontenc}
\usepackage{datetime}
\usepackage{fancybox}
\usepackage{graphicx}
\usepackage[usenames,dvipsnames]{pstricks}
\usepackage[linesnumbered,vlined,ruled]{algorithm2e}

\usepackage{amssymb}
\usepackage{amsmath}
\usepackage{amsthm}

\newtheorem{lemma}{Lemma}
\newtheorem{theorem}{Theorem}
\newtheorem{definition}{Definition}

\newcommand{\myskip}{\vskip 0.4cm\noindent}

\newcommand{\V}{\mathcal{V}}
\newcommand{\ovr}{\overline{r}}
\newcommand{\s}[1]{\hspace*{#1pt}}

\newenvironment{example}{\noindent\textbf{Example.}\\}{}

\newenvironment{fact}{\noindent\textbf{Fact.}\\}{}
\newenvironment{obs}{\noindent\textbf{Observation.}\\}{}

\journal{Discrete Applied Mathematics}
\newcommand\comment[1]{}

\begin{document}
\begin{frontmatter}

\title{
Compressed Pattern-Matching with Ranked Variables in Zimin Words
 }
 \author[umk]{Rados{\l}aw~G{\l}owinski\corref{cor1}}
\ead{glowir@mat.umk.pl}
 \author[uw,umk]{Wojciech Rytter\corref{cor1}\fnref{grant2}}
\ead{rytter@mimuw.edu.pl}

\cortext[cor1]{Corresponding author}
\fntext[grant2]{Supported by grant no. N206 566740 of the National Science Centre.}
\address[umk]{Faculty of Mathematics and Computer Science,\\
Nicolaus Copernicus University,\\
Chopina 12/18, 87-100 Toru{\'n}, Poland\\
(+48) (56) 611-3410, fax (+48) (56) 611-2987}

\address[uw]{Department of Mathematics, Computer Science and Mechanics, \\
    University of Warsaw,\\
    Banacha 2, 02-097 Warsaw, Poland\\
    fax: (+48)(22) 55-44-200\\
    glowir@mat.umk.pl, rytter@mimuw.edu.pl}

 \begin{abstract}
\vskip 0.2cm \noindent Zimin words are very special finite words which are
 closely related to the pattern-avoidability problem.
This problem  consists in testing
if  an instance of a given pattern with variables occurs in almost all words over any finite alphabet.
The problem is not well understood, no polynomial time algorithm is known and its
NP-hardness is also not known. The pattern-avoidability problem is
equivalent to searching for a  pattern (with variables) in a
Zimin word. The main difficulty is potentially exponential size of
 Zimin words. We use  special properties of Zimin words, especially that
 they are highly compressible,  to
design efficient algorithms for special version of the pattern-matching, called
here {\em ranked matching}.  It gives a new
interpretation of Zimin algorithm in compressed setting. We
discuss the structure of  rankings of variables and compressed
representations of values of variables. Moreover, for a ranked matching we present efficient algorithms to find
the shortest instance and the number of valuations of instances of the pattern. 
\end{abstract}

\end{frontmatter}

\section{Introduction}
The research on pattern avoidability started in late 70's in the
papers by Bean, Ehrenfeucht and McNaulty \cite{bean}, and
independently by Zimin  \cite{zimin}. In the avoidability problem
two disjoint finite alphabets, $A=\{a,b,c,...\}$ and
$V=\{x_1,x_2,x_3,...\}$ are given, the elements of $A$ are letters
(constants)
 and the elements of $V$ are variables. We denote the empty word by $\epsilon$. A pattern $\pi$  is a sequences of variables.
The language of a pattern with respect to an alphabet $A$ consists
of words  $h(\pi)$, where $h$ is any non-erasing morphism from
$V^*$ to $A^+$ . We say that word $w$ encounters pattern $\pi$ (or
pattern occurs in this word) when there exists a morphism $h$,
such that $h(\pi)$ is a subword of $w$. In the other case $w$
avoids $\pi$.

\myskip  The pattern $\pi$ is unavoidable on $A$ if every long
enough word over $A$ encounters $\pi$, otherwise it is avoidable
on $A$. If $\pi$ is unavoidable on every finite $A$ then $\pi$ is
said to be unavoidable. \myskip
\begin{example}
The pattern $\alpha\alpha$  is avoidable on a three letters alphabet,
 see \cite{thue}, however it is unavoidable on any smaller alphabet. However $\alpha\beta\alpha$ is unavoidable on any alphabet.
\end{example}
\myskip \noindent
\noindent The crucial role in avoidability problems play the words
introduced by  Zimin \cite{zimin}, called Zimin words, and denoted
here by $Z_k$. \myskip
\begin{definition} {\rm \textbf{(of Zimin words)}} Let $$Z_1=1,\; Z_k=Z_{k-1}\:k\:Z_{k-1}$$
\end{definition}
\myskip
\begin{example}
$Z_1=1, \hspace{5pt} Z_2=121, \hspace{5pt} Z_3=1213121, \hspace{5pt} Z_4=121312141213121$
\end{example}
\myskip Observe  that these words are exponentially long, however
they  have a very simple structure implying many useful
properties. Define the Zimin morphism  $$\mu(1)=121,\;
\ \textrm{and}\ (\forall\;i>1)\ \mu(i)= i+1\ . $$
\begin{fact}\mbox{ \ }
\begin{itemize}
\item
The morphism  $\mu$ generates next Zimin word by mapping each letter according to $\mu$.
In other words:\ $Z_{k}= \mu (Z_{k-1})$.
\item Each Zimin word, considered as a pattern, is unavoidable. Moreover it is a longest unavoidable pattern over k-th letter alphabet. There exists only one (up to letter permutation) unavoidable pattern of length $2^{k}-1$ over
a  k-th letter alphabet and it is $Z_k$.
    \end{itemize}
\end{fact}

\noindent The main property of Zimin words is that the avoidability problem is
reducible to pattern-matching in Zimin words, see  \cite{zimin}, \cite{heitsh}.
\myskip
\begin{lemma}\label{lemma5}
$\pi$ is an unavoidable pattern if and only if $\pi$ occurs in $Z_k$, where $k$ is the number of distinct symbols occurring in $\pi$.
\end{lemma}

\section{Compact representation of pattern instances}

\noindent
One of the basic proprties of Zimin words and their factors is related to the concept
of {\em interleaving} introduced below.

\begin{definition}  A sequence $u$ is $j$-interleaved iff for each two adjacent elements of  $u$ exactly oneequals $j$.
\end{definition}
\myskip
The Zimin word $Z_k$ can be alternatively defined as follows:
\begin{description}
\item{\bf (A)}\ $Z_k$ starts with 1 and ends with 1;
\item{\bf (B)}\  $|Z_k|=2^k-1$;
\item{\bf (C)}\  For each $1\le j\leq k$ after removing all elements smaller than $j$ the
obtained sequence is $j$-interleaved
\end{description}
\begin{lemma} \label{obs6}
If  $u\in \{1,2,\ldots,k\}^+$, then $u$ is a
factor of $Z_k$ iff it satisfies the condition (C).
\end{lemma}
\begin{proof}
We induction on $k$. For $k=1$ it is obvious.
Assume it is true for $j<k$ and $k>1$.

Remove all letters 1 from $u$, we obtain $u'$, then change each letter $i$ to $i-1$,
in this way we get a word $u''$, which also satisfies the condition (C).
By inductive assumption $u''$ is a factor of $Z_{k-1}$. We add letter 1 at the beginning and at the end of $\mu(u'')$ unless it is already there and denote obtained sequence as $v$. However $u$ is a factor
of $v$ and $\mu(Z_{k-1})=Z_k$.
Consequently $u$ is a factor of $Z_k$ as
a factor of $v$. This completes the proof.
\end{proof}

\noindent
This gives a simple linear time algorithm to check if an explicitly given sequence is a factor of $Z_k$. However we are dealing with patterns,
 and instances of the pattern can be exponential with respect to the length of the pattern and the number of distinct variables. The instance is given by values of each variable which are factors of $Z_k$.
  Hence we introduce compact representation of factors of Zimin words.
\myskip We partition $u$ into $w_1\:m\:w_2$, where m is a highest
number in $w$ (in every subword of Zimin word the highest number
occurs exactly once). Then for each element $i$ of $w_1$,
respectively $w_2$, we remove $i$ if there are larger elements to
the left and to the right of this element. In other words if there
is a factor $s\:\alpha\:i\:\beta\:t$, with $i<s,\:i<t$, we remove
the element $i$. Denote by $compress(u)$ the result of removing
all redundant $i$ in $u$.
\myskip
\begin{obs} \label{obs7}
$compress(u)$ uniquely encodes a factor of a Zimin word.
\end{obs}
\myskip
\begin{example}
$$compress(2141213121512131)\;=\;24531$$
$$ compress(Z_4)\;=\;compress(121312141213121)\;=\;1234321$$
\begin{figure}[ht]
\begin{center}
\begin{pspicture}(0,0)(8,1.5)
\definecolor{color34b}{rgb}{0.8509803921568627,0.8509803921568627,0.8509803921568627}
\definecolor{color34}{rgb}{0.996078431372549,0.996078431372549,0.996078431372549}
\definecolor{color32e}{rgb}{0.7098039215686275,0.7098039215686275,0.7098039215686275}
\definecolor{color32}{rgb}{0.00392156862745098,0.00392156862745098,0.00392156862745098}
\usefont{T1}{ptm}{m}{n}
\rput(0.4,2){{\Large $\alpha$}}\rput(2.2,2){{\Large $\beta$}}\rput(4.7,2){{\Large $\gamma$}}\rput(7.25,2){{\Large $\beta$}}
\rput(4,1.3){1 2 \psframebox[linewidth=0.04,framearc=0.5,linecolor=color34,fillstyle=solid,fillcolor=color34b]{\psframebox[linewidth=0.004,linecolor=color32,framearc=0.5,shadow=true,shadowcolor=color32e,fillstyle=solid,fillcolor=color34]{1 3 1 2} \psframebox[linewidth=0.004,linecolor=color32,framearc=0.5,shadow=true,shadowcolor=color32e,fillstyle=solid,fillcolor=color34]{1 4 1 2 1 3 1} \psframebox[linewidth=0.004,linecolor=color32,framearc=0.5,shadow=true,shadowcolor=color32e,fillstyle=solid,fillcolor=color34]{2 1 5 1 2 1 3 1 2} \psframebox[linewidth=0.004,linecolor=color32,framearc=0.5,shadow=true,shadowcolor=color32e,fillstyle=solid,fillcolor=color34]{1 4 1 2 1 3 1}} 2 1}
\rput(0.4,0.6){1 3 2 } \rput(2.2,0.6){1 4 3 1 } \rput(4.8,0.6){2 5 3 2 } \rput(7.3,0.6){1 4 3 1 }
\end{pspicture}

\caption{Example of a compact representation. In the first line there is a pattern, in the second uncompressed valuation of variables and in the third compressed valuation.}
\label{fig::ex1}
\end{center}
\end{figure}
\end{example}

\begin{fact}\label{fact3}
For any $u \in \{1,2,\ldots,k\}^+$ the first and the last elements of $u$ and $compress(u)$ are respectively equal.
\end{fact}
 Notice that compressed representation of any subword of $Z_k$ has
at most $2*k-1$ letters and the representation of all variables
requires $O(k^2)$ memory (under the assumption that only $O(1)$ space is necessary for representing each number). \myskip For a valuation (morphism) $h$
of the variables by its ranking function $R_{h}$ we mean the
function which assigns to each variable $x_i$ its rank:\ $R_h(x_i)$ denotes the maximal letter in $h(x_i)$.

\myskip
For a pattern $\pi$ and a given valuation of variable by $\pi_{(i)}$
we mean the pattern with variables of ranks smaller than $i$
removed from $\pi$.

\noindent By $\V_i(\pi)$ we define the set of variables from
$\pi$ with rank $i$.

\noindent For two strings $u,w$ we write $u\le w$ iff $u$ is a
subword of $w$. By $\pi\rightarrow_h Z_k$ we mean that $h(\pi)
\leq Z_k$ and by $\pi \rightarrow Z_k$ we mean that for some
morphism $h, h(\pi) \leq Z_k$.
%
\myskip
Now we present the algorithm to determine if the concatenation of compact representations $z_1, z_2, \dots , z_k$ is a Zimin subword. Let $m$ be the maximal number that occurs in any representation. For each $i$ we know that $z_i$ encodes some Zimin subword therefore we have to check connections between $z_i$ and $z_{i+1}$ for $i \in \{1,\dots,k-1\}$.

If there is exactly one occurrence of '1' between each two consecutive representations ('1' occurs at the end of the first representation or at the beginning of the second) then we remove all occurrences of '1' from all representations (it can only appear at the beginning or the end of a representation).  Similarly we check for occurrences of '2' and so on up to $m-1$. If during the process we obtain empty representation we remove it from further consideration. If we can't find such occurrences it means that the concatenation is not a Zimin subword. If there is only one representation left (others are empty) we know that this is a Zimin subword.
An idea of algorithm is given below. If we finish with the empty sequence then we accept the input.
\myskip
{\bf Reduction Algorithm}
\\
\hspace*{0.2cm}
{\bf while} the sequence is nonempty {\bf do}
\begin{enumerate}
\item check if there is exactly one occurrence of 'i' between representations 
of each two consecutive variables;
\item remove all occurrences of 'i' from all representations;
\item remove variables with empty representation;
\item increase $i$;
\end{enumerate}

\noindent Denote by $compress(\pi)$ the pattern
in which the  variables replaced by their compressed representations. Of course $|compress(\pi)|\le n\;k$, where $k$ is the maximal rank of a variable.
\begin{theorem}\mbox{ \ }\\
{\bf (a)}\
For $z_1, z_2, \dots , z_k$ with length $l$ we can check if concatenation $z_1 \cdot z_2 \cdot  \ldots \cdot  z_k \leq Z_m$ in time $O(l) $.\\
{\bf (b)} Assume we are given a pattern $\pi$ of size $n$ with $k$ variables
and a compact representation of values of the variables. Then we
can check if the given compressed instance of $\pi$ occurs in
$Z_k$ in time $O(compress(\pi))$.
\\
{\bf (c)}\ 
Assume we are given compressed representation of all variables.
If the pattern is a (unique) valid instance of a factor $y$ of $Z_n$ then we can compute
the compressed representation of $y$ in time $O(compress(\pi))$.
\end{theorem}
\begin{proof} 
The point (a) follows form an easy list implementation of the informal reduction algorithm. 
The point (b) follows directly from point (a).
The proof of point (c) works as follows. 
We rewrite the pattern as a sequence of compressed representation of
variables. For each element we compute the first larger element to the left and to the right. It corresponds to the first larger neighbor problem, known to be solvable in
linear time. The we start with leftmost element and generate the sequence $y'$ of first larger right neighbors.
Then we start with the right end of the whole sequence  and generate the sequence $y''$ of first larger left neighbors. The compressed representation of $y$ is the concatenation $y'\;y''$.
\end{proof}
\noindent We show now an alternative algorithm using extended compressed representations.
Each compressed representation can be treated as a concatenation of letters and smaller
Zimin words, for example the value of $2 5 3 1$ is $2\;Z_1\;5\;Z_2\;3\;Z_1\;2$.
It is called the {\em extended compressed representation}.
\
Generally, if the (previous) compressed representation is: 
$$z\ =\ i_1<i_2<\ldots i_j>i_{j+1}>i_{j+2}>\ldots i_r$$
then we change the boundary letters 1 to $Z_1$, if there are such letters, afterwards 
we change every $i_p$, for $i<j$ to $Z_{i_p-1}\; i_p$ and every $ j_p$ to $Z_{j_p-1}\;j_p$.
\
Denote by $Ext(x)$ the extended compressed representation of the variable $x$.
\begin{figure}[h]
\begin{center}
\includegraphics[width=4.7in]{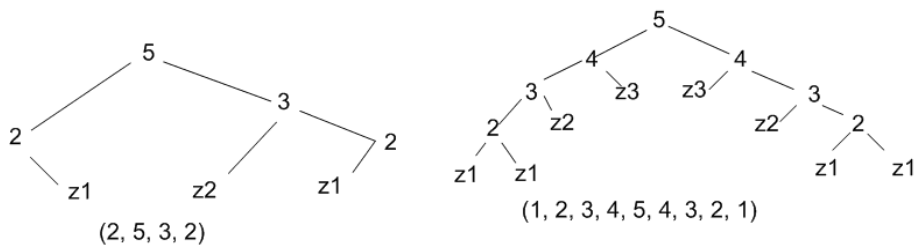}
\caption{Tree illustrations of two compressed representations. The second one corresponds to $ Z_5$.}\label{psc_rys1}
\end{center}
\end{figure}

\begin{example}
 We continue previous example. We have $$Ext(\alpha)=Ext(1 3 2)=Z_1\; 3\;  Z_1\;  2,\ \
 Ext(\beta)=Ext(1 4 3 1)=Z_1\;4\;Z_2\; 3\;  Z_1,$$
  $$Ext(\gamma)=Ext(2 5 3 2)=2\;  Z_1\;   5\;   Z_2\;  3\;   Z_1\;  2$$
 \end{example}

\noindent Let $Ext(\pi)$ be the extended compressed representation of the pattern $\pi$
resulting by substituting extended representations of variables. For example, with values of variables from the previous example, we have:
$$Ext(\alpha\; \beta\; \gamma\; \beta)\ = 
Z_1\; 3\;  Z_1\;  2\; Z_1\;4\;Z_2\; 3\;  Z_1\;2\;  Z_1\;   5\;   Z_2\;  3\;   Z_1\;  2\;  Z_1\;4\;Z_2\; 3\;  Z_1$$
We can scan this sequence from left to right and replace every factor 
$Z_{i-1}\;i\;Z_{i-1}\;i$ by $Z_i$
until no change are allowed. Finally we get the extended compressed representation
of the whole sequence, if it is a sequence corresponding to a factor of $Z_n$.
In our case after the first scan we get 
$$Z_1\; 3\;  Z_2\;4\;Z_2\; 3\;  Z_2\;   5\;   Z_2\;  3\;   Z_2\;4\;Z_2\; 3\;  Z_1$$
After the next scan we get:
$$Z_1\; 3\;  Z_2\;4\;Z_3\;   5\;   Z_3\;4\;Z_2\; 3\;  Z_1$$
which is an extended compressed representation of the factor of a Zimin word
corresponding to (non-extended) representation $(1,3,4,5,4,3,1)$.
Instead of scanning many time the extended pattern we can use a kind of a parsing tree,
which gives us the natural order of reductions. 

\begin{figure}[h]
\begin{center}
\includegraphics[width=4.7in]{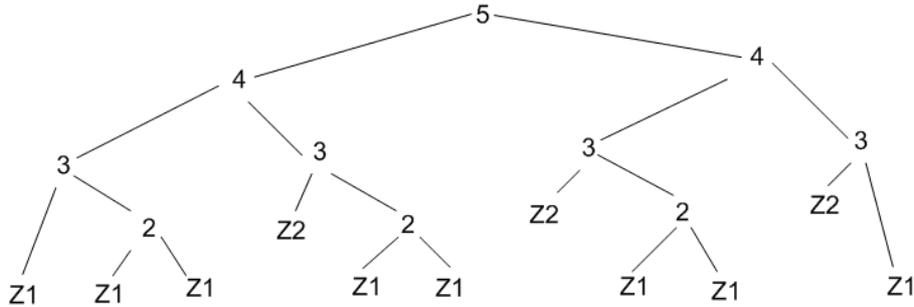}
\caption{The Cartesian tree is a natural parsing tree for a sequence of letters and smaller Zimin words, after the reduction we get extended compressed representation of the whole sequence, if its is a factor of $Z_n$.}
\end{center}
\end{figure}

\begin{figure}[h]
\begin{center}
\includegraphics[width=4.7in]{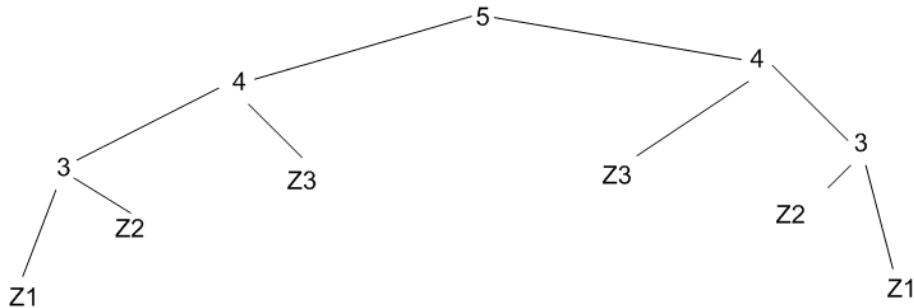}
\caption{After bottom-up reductions of the previous tree we get a representation of a factor of $Z_5$.}
\end{center}
\end{figure}

\noindent A Cartesian tree is a natural parsing
tree for extended compressed sequence.
The Cartesian tree is a useful data structure computable in linear time. For
a sequence of ordered elements we construct the tree as follows.
The root is the largest element. The left son is the largest element to the
left of the root, and its subtree is a Cartesian tree of the elements to the
left of the root. Similarly for the right part of the sequence.
\
In our case we assume that the priority of $Z_i$ is the same as the letter $i$.
The Cartesian tree corresponding to our sequence from the previosu example is shown below.


%
%
%
%
%
%
%

\section{Some properties of Zimin words and free sets}
 The
ranking sequence associated with $\pi$ is the sequence of ranks of
consecutive variables in $\pi$. \myskip
 Assume for a while that
our pattern $\pi$ is a permutation of $n$ variables and we ask for
the set of possible ranking sequences. \myskip The ranking
sequence has many useful properties:
\begin{enumerate}
\item Between every two occurrences of the same number $a$ in ranking sequence there should be a number larger
than $a$.
\item The ranking function is not necessarily  injective (one to one).
\item If $x_1x_2$ and $x_1x_3$ are subwords of a pattern $\pi$, $x_1,x_2,x_3 \in V$ and $rank(x_3)<rank(x_2)<rank(x_1)$ and there exists a morphism $\varphi $ that morphs $p$ into $Z_k$ then $\varphi(x_3)$ is a proper prefix of $\varphi(x_2)$.
\end{enumerate}
\myskip Let $\ovr(\pi,h)$ be the set of ranks of variables in $\pi$
for the valuation $h$. For example, for a pattern $\pi=\alpha \beta \alpha \gamma \beta \alpha$ and a valuation $h(\alpha)=1, h(\beta)=2, h(\gamma)=31$ ranking sequence is $(1, 2, 1, 3, 2, 1)$ and $\ovr(\pi,h)=\{1,2,3\}$.

 \myskip The following three facts are
consequences of the proof of Zimin theorem (see \cite{zimin} for
details). \myskip
\begin{lemma}\label{three}\mbox{ \ }
\begin{enumerate}
\item
If pattern $\pi$ is unavoidable then there exists a morphism $h$
such that $h(\pi)$ occurs in $Z_k$ and $\min \ovr(\pi,h)=1$.
\item
If  $\pi \rightarrow_h Z_k$ and $\min \ovr(\pi,h) = j+1 >1$ then there exists morphism $g$ such that $\pi \rightarrow_g Z_k$ and $\ovr(\pi,g)=\{r-j : r \in \ovr(\pi,h)\}$.
\item
If $\pi\rightarrow Z_k$ then there exists morphism $h$ such that the set of ranks is an
interval:\ $\ovr(\pi,h) = \{1, \ldots, m \}$, for some $1 \leq m
\leq k$.
\end{enumerate}
\end{lemma}
\myskip
\noindent We present Zimin algorithm based on free sets and
$\sigma$-deletions. \myskip
\begin{definition} \label{def13}
$F \subseteq V$ is a \textbf{free set} for $\pi \in V^+$ if and only if there exist sets $A,B \subseteq V$ such that $ F \subseteq B \setminus A$ where, for all $xy \leq \pi, x \in A$ if and only if $y \in B$.
\end{definition}
\myskip
\begin{definition}
The mapping $\sigma_F$ is a \textbf{$\sigma$-deletion} of $\pi$ if and only if $F\subseteq V$ is a free set for $\pi$ and $\sigma_F : V\rightarrow V \cup \{\epsilon\}$ is defined by

$$
\sigma_F(x) =
 \begin{cases}
 x &  \mbox{if } x \notin F \\
 \epsilon & \mbox{if } x\in F \\

 \end{cases}
$$
\end{definition}

The proof of the following fact can be found in \cite{algebraic}, Zimin's algorithm is based on this fact.
\myskip
\begin{lemma}
$\pi$ is an unavoidable pattern if and only if $\pi$ can be reduced
to $\epsilon$ by a sequence of $\sigma$-deletions.
\end{lemma}

\myskip
Unfortunately it is insufficient to remove only singleton free sets. There are patterns, which require the removing more than one element free sets,
 for example the pattern $$ \alpha\beta\alpha\gamma\alpha' \beta\alpha\gamma\alpha \beta\alpha'\gamma\alpha'\beta\alpha' $$
  Therefore we can have exponentially many choices for free sets.

\begin{lemma}
If $\pi\rightarrow Z_k$ then $R_1(\pi)$ is a free set.
\end{lemma}
\begin{proof}
To satisfy the definition of a free set we need to give sets $A$ and $B$, such that $R_1(\pi) \subset B \setminus A$ and all predecessors of variables from $B$ are in $A$, all successors of variables from $A$ are in $B$. We put all variables starting with $1$ as the set $B$ and all variables that do not end with $1$ as set $A$.
\end{proof}

\begin{lemma}\label{keylemma}
If $\pi \rightarrow Z_k$ then $\pi_{(2)} \rightarrow Z_{k-1}$.
\end{lemma}
\begin{proof}
If $\pi \rightarrow Z_k$ then there exists morphism $h$, such that $h(\pi)$ is a subword of $Z_k$. $\V_1(\pi)$ is a set of variables $x$, such that $h(x)=1$. We shall notice that if we remove all $1$ from $Z_k$ we obtain $Z_{k-1}$. We define a new morphism $g$ for all variables from $\V \setminus \V_1$ as $g(x)=f(h(x))$, where $f$ is a function that removes all occurrences of $1$ from a word. Now we will show that $g(\pi_{(2)})$ is a subword of $Z_{k-1}$. We see that $g(\pi_{(2)}) = f(h(\pi_{(2)})=f(h(\pi))$ because $\pi$ differs from $\pi_{(2)}$ only in variables that equal $1$. So occurrence $g(\pi_{(2)})$ equals $h(\pi)$ with all $1$ deleted and $g(\pi_{(2)})$ is a subword of $Z_{k-1}$.
\end{proof}
\begin{theorem}\label{theorem18}
A pattern $\pi$ occurs in $Z_k$ if and only if  $ \V_1(\pi)$ is a free set and $\pi_{(2)} \rightarrow Z_{k-1}$.
\end{theorem}
\begin{proof}
,,$\Rightarrow$''
This is a consequence of Lemma \ref{three}  and Lemma \ref{keylemma}.

\hspace{17pt},,$\Leftarrow$'' It follows from proof of Zimin theorem and can be found in \cite{zimin}.
\end{proof}

\section{Ranked pattern-matching}

It is not known (and rather unlikely true) if the pattern-matching
in Zimin words is solvable in polynomial time. We introduce the
following polynomially solvable version of this problem. \myskip
{\bf Compressed Ranked Pattern-Matching in Zimin Words:}
\begin{description}
\item{\bf Input:}\ given a pattern $\pi$ with $k$ variables and the
ranking function $R$

\item{\bf Output:}\ a compressed instance of an occurrence of $\pi$ in $Z_k$ with the given  ranking function,
or information that there is no such valuation, the values  of
variables are given in their compressed form
\end{description}

\noindent The algorithm for the ranked pattern matching can be
used as an auxiliary tool for pattern-matching without any ranking
function given. We can just consider all {\em sensible} ranking
functions. It gives an exponential algorithm since we do not know
what the rank sequence is. Although exponential, the set of {\em
sensible} ranking sequences can be usefully reduced due to special
properties of realizable rankings.

\subsection{Application of 2-SAT}
In one iteration we have not only to check if the set of variables of the smallest rank $i$
is a free set but we have to compute which of them start/end with a letter $i$.
However in a given iteration the letter $i$ can be treated as '1'.

In our algorithms we will use the function $FirstLast(\pi,W)$ which solves an instance of 2-SAT problem. It computes which variables should start-finish with
the smallest rank letter, under the assumption that variables from $W$ equal the smallest rank letter, to satisfy local properties of the Zimin word.
\myskip
In the function we can treat the smallest rank letter as $1$.
In Zimin word there are no two adjacent '1'. This leads to the
fact: for any adjacent variables $xy$ from $\pi$ either $x$ ends
with '1' or $y$ starts with '1'. If for a given pattern we know
that some variables start with '1' (or end with '1') we deduce
information about successors of this variable (or predecessors).
For example if we have a pattern $\beta\alpha\beta\gamma$ and we
know that $\alpha$ starts with '1' we deduce that $\beta$ does not
end with '1' and then deduce that $\gamma$ starts with '1'. For a
given set of variables that start and end with '1' we can deduce
information about all other variables in linear time (with respect
to the length of the pattern). \myskip
 For every variable $x$ from
the pattern we introduce two logic variables: $x^{first}$ is true iff
$x$ starts with '1', $x^{last}$ is true iff $x$ ends with '1'. Now
for any adjacent variables $xy$ we create disjunctions $x^{last}
\vee y^{first}$ and  $\neg x^{last} \vee \neg y^{first}$. If we
write the formula
\begin{align*}
 \begin{split}
 F= &(x_1^{last} \vee x_2^{first}) \wedge (\neg x_1^{last} \vee \neg x_2^{first}) \wedge (x_2^{last} \vee x_3^{first}) \wedge (\neg x_2^{last} \vee \neg x_3^{first}) \wedge \ldots  \\
&\ldots \wedge (x_{n-1}^{last} \vee x_n^{first}) \wedge (\neg
x_{n-1}^{last} \vee \neg x_n^{first})
\end{split}
\end{align*}
 we have an instance of 2-SAT problem.
For the variables $y_1, \ldots, y_s\in W$, that we know that valuate as '1',  we expand our formula to $F \wedge y_1^{first} \wedge y_1^{last} \wedge \ldots \wedge y_l^{first} \wedge y_s^{last}$.
\myskip
\begin{example}
If for a pattern $\beta\alpha\beta\gamma\alpha$ we know that valuation of $\alpha$ will be '1' we produce formula
\begin{align*}
 \begin{split}
 &( \beta^{first} \vee \alpha^{last} ) \wedge ( \neg \beta^{first} \vee \neg \alpha^{last} ) \wedge ( \alpha^{first} \vee \beta^{last} ) \wedge  \\
 \wedge &( \neg \alpha^{first} \vee \neg \beta^{last} ) \wedge ( \beta^{first} \vee \gamma^{last} ) \wedge
 ( \neg \beta^{first} \vee \neg \gamma^{last} )\wedge  \\
 \wedge &( \gamma^{first} \vee \alpha^{last} ) \wedge ( \neg \gamma^{first} \vee \neg \alpha^{last} ) \wedge ( \alpha^{first}) \wedge (\alpha^{last})
\end{split}
\end{align*}
In the general case there can be many solutions of the formula but in our example the only solution is: $\alpha^{first}=\alpha^{last}=\gamma^{first}=true, \hspace{3pt} \beta^{first}=\beta^{last}=\gamma^{last}=false$, which means that $\alpha $ starts and ends with '1', $\beta$ starts and ends with non-'1', $\gamma$ starts with '1' and ends with non-'1'.
\end{example}
A positive solution to this problem is necessary for the existence
of a valuation $val$ of the variables from pattern $\pi$, such that
$val(\pi) \leq Z_k$ and each variable $x$ starts (ends) with '1'
iff $x^{first}$ is true (resp. $x^{last}$ is true) in the
solution.
\myskip
It is well known that 2-SAT can be computed efficiently, consequently:

\begin{lemma}
We can execute $FirstLast(\pi,W)$ in linear time.
\end{lemma}
\subsection{The algorithm:  compressed and uncompressed versions}
Now we present two versions of the algorithm deciding if pattern $\pi$
occurs in Zimin word with given rank sequence and computing the values of variables, if there is
an occurrence. First of these algorithms uses
uncompressed valuations of variables and uses exponential space
and second one operates on compressed valuations. If the answer is
positive algorithms give valuations of variables, ie. morphism
$val$ such that $val(\pi) \leq Z_k$.

Denote by $alph(\pi)$ the set of symbols (variables) in $\pi$.  Let $\pi$ be the pattern with
given rank sequence which maximal rank equals $K$, $|\pi|=n,
|alph(\pi)|=k$. We define the operation $firstdel(i,s), lastdel(i,s)$ of
removing the first, last letter from $s$, respectively, if this letter is
$i$ (otherwise nothing happens), similarly define $firstinsert(i,s),
lastinsert(i,s)$: inserting letter $i$ at the beginning or at the end of $s$
if there is no $i$.

In general, maximal rank can occur multiple times, but in this case, because of properties mentioned earlier in this paper, the pattern with this rank sequence is avoidable and cannot occur in Zimin word. Both algorithms assume that there is only one occurrence of the maximal rank (we denote by $x_K$ the variable with the maximal rank).
\newpage

\begin{algorithm}[ht]
 $K:=$ maximal rank\;
 $\V_i$ is the set of variables of rank $i$, for $1\le i \le K$\;
$val(x_K):=1$\;
\medskip \For {  $i=K-1$ {\bf downto} 1} {
	\ForEach {$x\in \V_i$} {
		 $val(x):=1$\;	
		\If { $FirstLast(\pi_{(i)},\V_i)$} {
			 {\bf comment}: {\it we know now which variables in $\pi_{(i)}$ start/finish\\
			 \hspace{49pt} with $i$ due to evaluation of a corresponding 2SAT}\\
			\ForEach { $x\in \V_{i+1}\cup \V_{i+2} \cup \ldots \cup\V_K$  } {
		$val(x):=\mu(val(x))$\;
		\lIf {$not(x^{first})$} {$val(x):= firstdel(1,val(x))$\;}
		\lIf {$not(x^{last})$} {$val(x):= lastdel(1,val(x))$\;}
		\lIf {$x^{first}$} {$val(x):= firstinsert(1,val(x))$\;}
		\lIf {$x^{last}$} {$val(x):= lastinsert(1,val(x))$\;}
		}
	}
	\lElse {\Return false};
}
}
  \label{uncompressedembedding}
 \caption{\sl Uncompressed-Embedding$(\pi,Z_K)$}
\end{algorithm}

The next algorithm is a space-efficient simulation of the
previous one.

We are not using the morphism $\mu$, instead of that the
values of variables are maintained in a compressed form, we are
adding to the left/right decreasing sequence of integers.

\begin{algorithm}[ht]
$K:=$ maximal rank\;
$\V_i$ is the set of variables of rank $i$\;
\medskip
 \For {$1\le i \le K$} {
 $val(x_K):=K$\;
 }
 \medskip
 \For {$i=K-1$ {\bf downto} 1} {
	 \ForEach { $x\in \V_i$} {
		 $val(x):=k$\;
		\If {$FirstLast(\pi_{(i)},\V_i)$} {
			\ForEach {$x\in \V_{i+1}\cup \V_{i+2} \cup \ldots \cup \V_K$}{
				\lIf { $x^{first}$} {$val(x):=firstinsert(i,val(x))$\;}
				\lIf { $x^{last}$ } {$val(x):= lastinsert(i,val(x))$\;}
				\lElse {\Return false;}
			}
		}
	}
	
}
\caption{\sl Compressed-Embedding$(\pi,Z_K)$}
\label{alg:compressedembedding}
\end{algorithm}

{\bf Example}
\
 Below we present an example of the Uncompressed-Embedding algorithm for
 $$ \pi=\s{10} \delta \s{3} \alpha \s{3} \gamma
\s{3} \beta \s{3} \lambda  \s{3} \gamma \s{3} \alpha \s{3} \delta
\s{3} \alpha \s{3} \gamma \s{3} \beta \s{3} \alpha$$
\s{63} with the rank sequence $\s{25}4 \s{4} {1} \hspace{4pt} {3} \s{4} {2} \s{4}{5} \s{4} {3} \s{4}{1} \s{4} {4} \s{4}{1} \s{4} {3} \s{4} {2} \s{4} {1} \s{5}$

\vspace{6pt}
\begin{tabular}{c}
$\s{8} \lambda_\downarrow $ \\
$\s{1} \overbrace{1} $\\
 First we set the variable with the highest rank.   $val(\lambda)=1$ \\
           \hline
           \\
           $\delta_\downarrow  \lambda \delta_\downarrow$\\
           $ \overbrace{121} $ \\
           $i=4$. We set $val(\delta)=1$ and morph $val(\lambda)=121$. \\
           From solution of 2-SAT we know that $\delta$ starts and ends with '1',\\
           $\lambda$ starts and ends with non-'1', we set $val(\lambda)=2$.\\
\hline
         $  \s{16} \delta \gamma_\downarrow \s{3}\lambda   \gamma_\downarrow \s{7}\delta \gamma_\downarrow   $\\
         $  \overbrace{121}\overbrace{3}\overbrace{121} $  \\
         $i=3$. We have $val(\gamma)=1, val(\delta)=121, val(\lambda)=3$.\\
         From solution of 2-SAT: $\gamma$ starts and ends with '1', \\
         $\delta$ and $\lambda$ start and end with non-'1',\\
         therefore we set $val(\lambda)=3$, $val(\delta)=2$\\
\hline
          $\s{31} \delta \s{14} \gamma \s{3} \beta_\downarrow \s{3}\lambda \s{15} \gamma \s{17}\delta \s{13} \gamma \beta_\downarrow  $\\
          $ 121 \overbrace{3}\overbrace{121}\overbrace{4}\overbrace{121} \overbrace{3}\overbrace{121} $ \\
            $i=2$ : $val(\beta)=1, val(\gamma)=121, val(\delta)=3, val(\lambda)=4$.\\
            From solution of 2-SAT: $\beta$ starts and ends with '1', $\gamma,\delta$ start with '1',  \\
            end with non-'1', $\lambda$ starts and ends with non='1' \\
            Therefore $val(\gamma)=12, val(\delta)=13, val(\lambda)=4$. \\
\hline
           $ \s{36} \delta \s{2} \alpha_\downarrow \s{5} \gamma \s{15} \beta \s{15}\lambda \s{15} \gamma    \s{2} \alpha_\downarrow \s{3}\delta \s{3} \alpha_\downarrow \s{2} \gamma \s{18} \beta \alpha_\downarrow  $\\
 $1213\overbrace{1214}\overbrace{1213}\overbrace{121}\overbrace{5}\overbrace{1213}\overbrace{1214}\overbrace{1213}
 \overbrace{121} $ \\
              $i=1$ : $val(\alpha)=1, val(\beta)=121, val(\gamma)=1213,$\\
            $val(\delta)=1214, val(\lambda)=5$.  \\
 From solution of 2-SAT: $\alpha,\lambda$ start and ends with '1',  \\
 $\beta$ starts with 1, ends with non-'1', $\gamma,\delta$ start and end with non-'1'. \\

           \end{tabular} \\

\vspace{10pt} \noindent Finally we have valuation of variables,
such that $val(\pi) \leq Z_5$.

            $\s{20} \delta \s{14} \alpha \s{13} \gamma \s{14} \beta \s{13} \lambda \s{14} \gamma \s{13} \alpha \s{14} \delta \s{13} \alpha \s{14} \gamma \s{13} \beta \s{14} \alpha $\\
            $12131\overbrace{214}\overbrace{1}\overbrace{213}\overbrace{12}\overbrace{151}\overbrace{213}\overbrace{1}\overbrace{214}\overbrace{1}\overbrace{213}\overbrace{12}\overbrace{1}$
\myskip
{\bf Example}
\myskip Now we present an example of
the $Compressed-Embedding$ algorithm for a different pattern. Now the function $val$ will be a compact representation instead of the full representation used in the last algorithm. We take:
$$\pi\ =\ \s{25} \alpha \s{3} \gamma \s{3} \beta \s{3} \delta \s{3} \eta \s{3} \alpha \s{3} \gamma \s{3} \beta \s{3} \zeta \s{3} \eta \s{3} \alpha$$
and the ranking sequence
$$\hspace*{1.8cm} \s{1}1 \s{4} 3 \s{4} 2 \s{4} 4 \s{3} 5 \s{3} 1 \s{4} 3 \s{4} 2 \s{3} 6 \s{4} 5 \s{3} 1$$
\myskip
\begin{tabular}{c}
 First we set the variable with the highest rank.   $val(\zeta)=6$ \\
           \hline
           \\
           For $i=5$ we have $\pi_{(5)}=\eta\zeta\eta$\\
           We set $val(\eta)=5$ and solve 2-SAT for \\
           $F=(\eta^{last} \vee \zeta^{first}) \wedge (\neg \eta^{last} \vee \neg \zeta^{first}) \wedge (\zeta^{last} \vee \eta^{first}) \wedge $ \\ $\wedge (\neg \zeta^{last} \vee \neg \eta^{first}) \wedge (\eta^{first}) \wedge (\eta^{last})$.\\
           From $FirstLast(\pi_{(5)},{\eta})$ we know that $\zeta_{first}$ and $\zeta_{last}$ are false,\\
           therefore we don't change $val(\zeta)$\\
\hline
         $i=4$. $\pi_{(4)}=\delta\eta\zeta\eta$\\
         We set $val(\delta)=4$ and execute $FirstLast(\pi_{(4)},{\delta})$. There are two solutions, \\
         we choose one of them and obtain $\eta^{first}=\eta^{last}=0$ and $\zeta^{first}=\zeta^{last}=1$,\\
         therefore we change $val(\zeta)=464$.\\
\hline
            $i=3$ : $\pi_{(3)}=\gamma\delta\eta\gamma\zeta\eta$.\\
            We set $val(\gamma)=3$ and execute $FirstLast(\pi_{(3)},{\gamma})$. We know that\\
            $\delta^{first}=\delta^{last}=0,\s{2}\eta^{first}=1,\s{2}\eta^{last}=0$ and $\zeta^{first}=\zeta^{last}=0$,  \\
            therefore we only add $3$ at the beginning of $val(\eta)$, ie. $val(\eta)=35$\\
\hline
              $i=2$ : $\pi_{(2)}=\gamma\beta\delta\eta\gamma\beta\zeta\eta$.\\
            We set $val(\beta)=2$ and execute $FirstLast(\pi_{(2)},{\beta})$. We know that\\
 $\eta^{first}=\eta^{last}=1$ and rest of logic variables equal $0$.  \\
 We only add $2$ at the beginning and end of $val(\eta)$ ($val(\eta)=2352$)\\
\hline
              $i=1$ : $\pi_{(1)}=\pi=\alpha\gamma\beta\delta\eta\alpha\gamma\beta\zeta\eta\alpha$.\\
            We set $val(\alpha)=1$ and execute $FirstLast(\pi_{(1)},{\alpha})$. We know that\\
 $\gamma^{last}=\delta^{first}=\eta^{first}=\zeta^{first}=1$ and rest of logic variables equal $0$.  \\
 We add $1$ at the beginning of $\delta,\eta,\zeta$ and at the end of $\gamma$. \\
 We change: $val(\gamma)=31,val(\delta)=14,val(\eta)=12352,val(\zeta)=1464$. \\
           \end{tabular} \\

\vspace{10pt} \noindent Finally we have the compressed valuation of the variables (below we show full representations):
\myskip
\begin{figure}[ht]
\begin{center}
\begin{pspicture}(0,0)(12,2)
\definecolor{color34b}{rgb}{0.8509803921568627,0.8509803921568627,0.8509803921568627}
\definecolor{color34}{rgb}{0.996078431372549,0.996078431372549,0.996078431372549}
\definecolor{color32e}{rgb}{0.7098039215686275,0.7098039215686275,0.7098039215686275}
\definecolor{color32}{rgb}{0.00392156862745098,0.00392156862745098,0.00392156862745098}
\usefont{T1}{ptm}{m}{n}
\rput(0,2){{\Large $\alpha$}}\rput(1.2,2){{\Large $\beta$}}\rput(2.7,2){{\Large $\gamma$}}\rput(4.25,2){{\Large $\delta$}}\rput(6.25,2){{\Large $\eta$}}\rput(9.25,2){{\Large $\zeta$}}
\rput(0, 1.3){{\psframebox[linewidth=0.004,linecolor=color32,framearc=0.5,shadow=true,shadowcolor=color32e,fillstyle=solid,fillcolor=color34]{1}}
\rput(0.9,0.11){ \psframebox[linewidth=0.004,linecolor=color32,framearc=0.5,shadow=true,shadowcolor=color32e,fillstyle=solid,fillcolor=color34]{2}}
\rput(2.5,0.11){ \psframebox[linewidth=0.004,linecolor=color32,framearc=0.5,shadow=true,shadowcolor=color32e,fillstyle=solid,fillcolor=color34]{3 1}}
\rput(4,0.11){ \psframebox[linewidth=0.004,linecolor=color32,framearc=0.5,shadow=true,shadowcolor=color32e,fillstyle=solid,fillcolor=color34]{1 4}}
\rput(6,0.11){ \psframebox[linewidth=0.004,linecolor=color32,framearc=0.5,shadow=true,shadowcolor=color32e,fillstyle=solid,fillcolor=color34]{1 2 3 5 2}}
\rput(9,0.11){ \psframebox[linewidth=0.004,linecolor=color32,framearc=0.5,shadow=true,shadowcolor=color32e,fillstyle=solid,fillcolor=color34]{1 4 6 4}}
}

\rput(0,0.6){1} \rput(1.2,0.6){2} \rput(2.8,0.6){3 1} \rput(4.3,0.6){1 4} \rput(6.4,0.6){13121512} \rput(9.4,0.6){
141213121612131214}
\end{pspicture}
\end{center}
\end{figure}
\myskip
\noindent Our algorithms rely on the following lemma.
\myskip
\begin{lemma}\label{lemma20}
Let $i\in \{1,\ldots,k\}$ and the pattern $\pi_{(i+1)}$ occurs in $Z_{k-i}$. Pattern $\pi_{(i)}$ (equal $\pi_{(i+1)}$ with additional variables from $\V_i$) occurs in $Z_{k-i+1}$ iff the corresponding 2-SAT problem has a solution. Moreover every immersion of $\pi_{(i)}$, such that valuations of variables from $\V_i$ equal '1', corresponds to the solution of 2-SAT problem, such that valuation of every variable satisfies logic constraints on first and last character.
\end{lemma}
\begin{proof}
$\s{5}''\Leftarrow''$\\
First we observe that solution of 2-SAT problem guarantees that $\V_i$ is a free set. We put $A=\{v \in \pi_{(i)}: v^{last}=0\} $ and $B=\{v \in \pi_{(i)}: v^{first}=1\} $, which satisfy Definition \ref{def13}.
From Theorem \ref{theorem18} $\pi_{(i)}$ occurs in $Z_{k-i+1}$. Now we use Zimin morphism $\mu$ on variables from $\pi_{(i+1)}$ and set every variable from $\V_i$ to '1', then we modify (adding or removing '1' at the beginning or end) every valuation accordingly to logic constraints from the 2-SAT solution. Similarly as in the Zimin theorem proof (\cite{zimin}) we see that properties of $A$ and $B$ guarantee that modified valuations concatenate into proper immersion in $Z_{k-i+1}$.
\\$\s{5}''\Rightarrow''$\\
We have immersion of $\pi_{(i)}$ in $Z_{k-i+1}$ with valuation $val(v)$, such that for $v\in\V_i$ $val(v)=1$. We give a solution to 2-SAT problem as follows: for every variable $v$ from $\pi_{(i)}$ we set $v^{first}=1$ iff $val(v)$ starts with '1', $v^{last}=1$ iff $val(v)$ ends with '1'. Because Zimin word is $1$-interleaved this solution is correct.
\end{proof}
\myskip
\begin{theorem}
The compressed ranked pattern matching in Zimin words  can be
solved in time $O(n*k)$ and (simultaneously) space $O(n+k^2)$,
where $n$ is the size of the pattern and $k$ is the highest rank
of a variable. A compressed instance of the pattern can be
constructed within the same complexities, if there is any
solution.
\end{theorem}
\begin{proof}
We use the $Compressed-Embedding$ algorithm. First we embed $\pi_{(k)}$ into $Z_1$. Then we check for  a subsequent $i$, $k-1\geq i \geq 1$, if it is possible to embed $\pi_{(i)}$ having  embedding of $\pi_{(i+1)}$ using suitable 2-SAT and Lemma \ref{lemma20}. Non-existence of immersion of $\pi_{(i)}$ implies that whole $\pi$ does not occur into $Z_k$. Otherwise we get compressed valuation of variables, such that $val(\pi) \leq Z_k$.

We will consider time complexity of the Compressed-Embedding algorithm.
Because $|V_1| \cup |V_2| \cup \ldots \cup |V_k|=k$ first ${\bf for \ each}$ loop
executes exactly $k$ times during whole execution.

We solve 2-SAT problem exactly $k-1$ times for $\pi_{(i)}$ (length of formula is linear with respect to $|\pi_{(i)}|$) and for every $i$ $|\pi_{(i)}| \leq |\pi| =n$. That gives complexity $O(n*k)$ for this step.

We execute second {\bf for each} loop $k-1$ times. In each iteration we have:
$$|\V_{i+1}\cup \V_{i+2} \cup \ldots \cup \V_K| \leq |\V_{1}\cup \V_{2} \cup \ldots \cup \V_K| = k$$
\noindent Hence complexity for this step is $O(k^2)$.

\noindent
Finally, the algorithm has time complexity $O(k + n*k + k^2)=O(n*k)$, because $k \leq n$.

We need to remember the pattern of size $n$ and actual compact representations of the variables ($O(k^2)$). In each step we generate 2-SAT formula which has size $O(n)$. The algorithm needs $O(n+k^2)$ memory.
\end{proof}

\section{Different valuations}
We can obtain different valuations of the variables by choosing different solutions to our 2-SAT formula.

 We change the way of deciding which variable valuation starts or ends with '1' because finding all solutions to 2-SAT problem is in $\#P$. We will use adjacency graph $AG(\pi)$, i.e. graph with vertex set containing two copies of variables set. For each variable $v$ there are two vertices in $AG$ - $v^L$ and $v^R$. There is and edge between $v^L$ and $w^R$ if and only if $vw$ is a subword of $\pi$. This is a bipartite graph, one part contains elements of the form $x^L$ and the other elements $x^R$.

We will valuate vertices if $AG$ with $1$ or $0$. Variable $v^L$ will have value $1$ iff. valuation of $v$ ends with '1'. Variable $v^R$ will have value $1$ iff. valuation of $v$ starts with '1'. Whenever we know a vertex value we can give values to all vertices from the same connected component (because Zimin word is 1-interleaved).

\subsection{Shortest instance}

Let $\pi$ be an unavoidable pattern with given rank sequence. $\pi$ occurs in $Z_K$ because it is unavoidable. By choosing shorter valuations for some variables during the algorithm we can generate such valuation $val$, that $val(\pi)$ is the shortest among all possible valuations.

After i-th step of algorithm we have valuation $val$ of $\pi_{(i+1)}$ such that $val(\pi_{(i+1)})\leq Z_{i-1}$. Then we morph valuations with morphism $\mu$ ($val:=\mu(val)$) and obtain $val(\pi_{(i+1)})\leq Z_{i}$. Finally we try to insert variable with rank $i$ and manipulate starting and ending '1'. It is easy to see that to obtain shortest $val(\pi_{(i)})$ we need to try to remove left '1' from valuation of $first(\pi_{(i)})$ (first variable) and right '1' from $last(\pi_{(i)})$ (last variable). Manipulating starting and ending '1' of variables that are in the inside $\pi_{(i)}$ (that are not first or last variable) will not change the length of whole valuation. That is because if we remove '1' from one variable it will be added to its neighbour and length will be the same.

Let $v$ be the first variable of $\pi_{(i)}$. If rank of $v$ is $i$ (it is the variable that we add in this step) then we set it's valuation to '1' and cannot manipulate. If rank of $v$ is larger than $i$ and $v^L$ is not in the connected component of newly add variable then we can set $v^L=0$ (then it's valuation will not start with '1') and our valuation of  $\pi_{(i)}$ will be shorter.

Analogously we can deal with last variable of $\pi_{(i)}$ - $w$. In this case we have to check if $w^R$ is not in the connected component of newly added variable and set $w^R=0$ (so it will not end with '1').

By managing first and last variable during the algorithm execution we can obtain the shortest valuation of $\pi$.

To update the algorithm we only need to change function FirstLast. It solved the 2SAT formula to find out which variable valuations start or end with '1'. Now we will use adjacency graph $AG$.

Let $\pi$ be a pattern and $W$ set of variables, $AG$ be the adjacency graph for $\pi$, $v$ the first letter of $\pi$, and $w$ the last letter. We define $W^R=\{ x^R ; x \in W \}$ and $W^L=\{ x^L ; x \in W \}$. Function $valuate(x,i), x\in AG, i\in \{ 0,1\}$ sets a proper value of all vertices in the connected component of $x$. If element $x$ already has a value different than $i$ then function returns false. In the other case $valuate(x,i)$ sets value of elements in the connected component of x: every element in the same part (of bipartite graph) as $x$ is set to $i$ and every element in the opposite part is set to the opposite value. The time complexity of $valuate(x,i)$ is $O(l)$, where $l$ is the number of elements in the connected component of $x$.
\begin{algorithm}[ht]
Generate adjacency graph $AG$ for $\pi$\;
\smallskip
\ForEach {$x \in W$} {
	 \lIf {{\bf not} $valuate(x^L,1)$} {\Return false\;}
	 \lIf {{\bf not} $valuate(x^R,1)$} {\Return false\;}
}
\smallskip
\lIf {$v^L$ has no value} {$valuate (v^L,0)$\;}
\smallskip
\lIf {$w^R$ has no value} {$valuate (w^R,0)$\;}
\smallskip
\ForEach {$x^L$ in $AG$ that has no value} {
		$valuate(x^L,0$)\;
	}

\Return {true};
\caption{\sl ShortestFirstLast$(\pi,W)$}
\label{alg:cubic-runs}
\end{algorithm}

\begin{theorem}
For a given  ranked pattern  the compressed shortest instance
(if there is any instance) and
the number of instances of the ranked pattern occurring in  $Z_k$
can be constructed in time $O(n*k)$ and (simultaneously) space
$O(n+k^2)$.
\end{theorem}
\begin{proof}
To construct compressed shortest instance we need to use $ShortestFirstLast$ instead of $FirstLast$ in $Compressed-Embedding$ algorithm. To count the number of instances we need to know how many connected components are not valuated after valuating the $AG$ with newly added variables in each step. Each unvaluated connected component gives us two possibilities. If after last step of algorithm we know that there were $l$ unvaluated components in all steps then there are $2^l$ instances.
 
We will consider the time complexity. We will show that we are able to run $ShortestFirstLast$ in time $O(n)$ (same as the original $FirstLast$).

We generate adjacency graph for the subsequence of $\pi$, it has at most $2*k$ vertices and $n-1$ edges, therefore it takes $O(n)$ time. Function $valuate$ visits or sets value to each element of $AG$ exactly once during the whole algorithm, because during each call it processes one connected component. Total time of all $valuate$ calls is equal to time needed to generate all connected components (which is linear).

Finally time complexity of $ShortestFirstLast$ is $O(n)$.

The $AG$ has size $O(k^2)$ therefore whole $Compressed-Embedding$ algorithm uses $O(n+k^2)$ memory.
\end{proof}

\bibliographystyle{model1-num-names}
\bibliography{glowirels}
\end{document}